\pdfoutput=1
\documentclass[11pt]{article}
\usepackage[utf8]{inputenc}
\usepackage[toc,page]{appendix}
\usepackage{amsmath,amssymb,hyperref,array,xcolor,multicol,verbatim,mathpazo,fullpage}
\usepackage{amsthm}
\usepackage[normalem]{ulem}
\usepackage{float}
\usepackage{subcaption}
\usepackage[pdftex]{graphicx}
\usepackage{fancyhdr}
\usepackage{amsmath}
\newtheorem{theorem}{Theorem}[section]
\newtheorem{lemma}[theorem]{Lemma}

\theoremstyle{definition}
\newtheorem{definition}[theorem]{Definition}

\newtheorem{claim}[theorem]{Claim}
\newcommand{\cev}[1]{\reflectbox{\ensuremath{\vec{\reflectbox{\ensuremath{#1}}}}}}
\usepackage{listings}
\usepackage[T1]{fontenc}
\usepackage[numbers,sort&compress]{natbib} 
\usepackage[utf8]{inputenc}
\usepackage{authblk}

\begin{document}
\title{Network coding in undirected graphs is either very helpful or not helpful at all}
\author[1]{Mark Braverman \thanks{mbraverm@princeton.edu. Supported in part by an NSF CAREER award (CCF-1149888), NSF CCF-1525342, a Packard Fellowship in Science and Engineering, and the Simons Collaboration on Algorithms and Geometry.}}
\author[1]{Sumegha Garg \thanks{sumeghag@cs.princeton.edu}}
\author[1]{Ariel Schvartzman \thanks{acohenca@cs.princeton.edu}}
\affil[1]{Department of Computer Science, Princeton University}
\maketitle

\begin{abstract}

While it is known that using network coding can significantly improve the throughput of directed networks, it is a notorious open problem whether 
coding yields {\em any} advantage over the multicommodity flow (MCF) rate in undirected networks.  It was conjectured in~\cite{LiLiConjecture} that the answer is \lq no\rq. In this paper we show that even a small advantage over MCF can be amplified to yield a near-maximum possible gap. 

We prove that any undirected network with $k$ source-sink pairs that exhibits a $(1+\varepsilon)$ gap between its MCF rate and its network coding rate can be used to construct a family of graphs $G'$ whose gap is $\log(|G'|)^c$ for some constant $c < 1$. The resulting gap is close to the best currently known upper bound, $\log(|G'|)$, which follows from the connection between MCF and sparsest cuts. 

Our construction relies on a gap-amplifying graph tensor product that, given two graphs $G_1,G_2$ with small gaps, creates another graph $G$ with a gap that is equal to the product of the previous two, at the cost of increasing the size of the graph. We iterate this process to obtain a gap of $\log(|G'|)^c$ from any initial gap. 

\end{abstract}
\newpage   
\section{Introduction}
\label{sec:intro}

The area of network coding addresses the following basic problem: in a distributed communication scenario, 
can one use coding to outperform packet routing-based solutions? While the problem of communicating information 
over a network can be viewed as the process of moving information packets between terminals, a key distinction 
between moving packets and moving commodities is that information packets can be re-encoded by intermediate nodes.
For example, a node which receives packets $P_1$ and $P_2$ can calculate and transmit the bitwise XOR packet $P_1\oplus P_2$ 
to its neighbor. This operation has no analogue in multicommodity flow scenarios.

Whether (and to what extent) this ability confers any benefits over the simple routing-based solution, depends on the specific goal 
of the communication at hand. Such goals may include uni-cast and multi-cast throughput, error-resilience and security, to name a few. 
These questions have been the subject of active study in the recent past. A summary of major directions can be found in the books \cite{yeung2008, medard2012}, 
and surveys \cite{Fragouli2007, Yeung2005}. 

In this paper we focus on noiseless unicast communication. The network is a capacitated graph $G$ with $k$ source-sink terminal 
pairs $(s_i,t_i)$. Each source  vertex $s_i$ wants to transmit an information stream to $t_i$. The network coding rate NC$(G)$ is 
the maximum rate at which transmission between all pairs can happen simultaneously, given the capacity constraints. 

 If we forbid coding, and restrict 
nodes to forwarding information packets that they receive, the problem becomes equivalent to multicommodity flow over $G$ --- the very well-studied 
problem of maximizing the rate MCF$(G)$ at which commodities are moved from sources to sinks subject to the capacity constraints  (see e.g. \cite{AhujaRMO} for background). Clearly, the multicommodity rate can always be achieved --- but can it be beaten using ``bit tricks''?

If the graph $G$ is directed, there are well-known examples which show that coding can improve throughput in a very dramatic way
 \cite{KleinbergHR04,Kleinberg:Networks}. There is a family of examples $G$, where the gap between the multicommodity flow rate and 
network coding throughput is as large as $O(|G|)$. Despite substantial effort, it is not clear whether coding confers {\em any} 
benefit over routing in undirected networks. Li and Li \cite{LiLiConjecture} conjectured that the answer is `no'. This conjecture is currently open. 

It is known that the Li and Li conjecture holds in some special cases. Naturally, it holds whenever the sparsity of the graph matches the multicommodity flow rate. For cases where these quantities are not equal, \cite{Jain2006} and \cite{Kramer2006} show that the conjecture is true for the Okamura-Seymour graph and \cite{Jain2006, Kleinberg:Networks} show it for an infinite family of bipartite graphs. Empirical evidence also suggests that the conjecture is true \cite{LiLiJLau05}. 

One simple case where coding rate cannot exceed capacity is the case when the channel is a single edge: two parties cannot be sending messages to each other at a total rate exceeding the channel's capacity. This is a simple consequence of Shannon's Noiseless Coding Theorem. As a simple corollary, the sparsest cut in $G$ provides an upper bound for the network coding rate NC$(G)$. The sparsity of a cut $(U,V\setminus U)$ is defined 
as 
\begin{equation}\label{eq:1}
\text{Sparsity}(U,V\setminus U):=\frac{\text{Capacity}(U,V\setminus U)}{\text{Demand}(U,V\setminus U)}
\end{equation}
If we merge the vertices on either side of the cut, the network coding rate becomes $
\text{Sparsity}(U,V\setminus U)$. Merging nodes can only increase network coding rate, and thus we have
NC$(G)\le\text{Sparsity}(U,V\setminus U)$. Since the sparsity of $G$ is defined as the minimum of \eqref{eq:1} 
over all cuts, we have NC$(G)\le \text{Sparsity}(G)$. 

As discussed below, the multicommodity flow problem is very well-studied. In the one commodity case, the Max-Flow Min-Cut Theorem asserts 
that sparsity is equal to the flow rate. In the multicommodity case, the sparsity is still an upper bound on the multicommodity flow rate 
MCF$(G)$, but it might be loose by a factor of $\log |G|$:
\begin{equation}\label{eq:2}
 \text{Sparsity}(G)/O(\log |G|) \le \text{MCF}(G) \le  \text{Sparsity}(G).
\end{equation}
Thus the advantage one can gain for network coding over undirected graphs is at most $O(\log|G|)$:
\begin{equation}\label{eq:3}
 \text{NC}(G)/O(\log|G|)\le \text{Sparsity}(G)/O(\log |G|) \le \text{MCF}(G) \le \text{NC}(G).
\end{equation}
The Li and Li conjecture asserts that the rightmost `$\le$' is indeed an equality. Our main result is that either 
the conjecture is true, or it must be nearly `completely false': the gap between NC$(G)$ and MCF$(G)$ 
can be as high as poly-logarithmic in $|G|$. 


\begin{theorem}
\label{thm:Main1}
Given a graph $G$ that achieves a gap of $1+\epsilon$ between the multicommodity flow rate and the network coding rate, we can construct an infinite family of graphs $\widetilde{G}$ that achieve a gap of $O\left(\log|\widetilde{G}|\right)^c$ for some constant $c < 1$ that depends on the original graph $G$. 
\end{theorem}

In order to prove Theorem~\ref{thm:Main1}, we will show a simpler construction that can be applied repeatedly.
 
\begin{theorem}
\label{thm:Main2}
Given a graph $G$ of size $n$ with a gap of $1+\epsilon$ between the multicommodity flow rate and the network coding rate, we can create another graph $G'$ of size $n^{c^2}$ and a gap of $(1+\epsilon)^{2}$, where $c$ depends on the diameter of the graph $G$. 
\end{theorem}

\paragraph{Proof outline.} The main part of the construction is to define a \emph{graph tensor} on graphs $G_1$ and $G_2$ that have gaps of $1+\epsilon_1$ and $1+\epsilon_2$ respectively, between the multicommodity flow rate and the network coding rate, which gives a new graph $G$ with a gap of $(1+\epsilon_1)(1+\epsilon_2)$  while keeping a check on the size of $G$. We can then take a graph with a small gap and tensor it with itself to produce a graph with an even larger gap. 
Repeatedly tensoring the output of the previous iteration with itself will give us Theorem~\ref{thm:Main1}.


\begin{figure}
\begin{minipage}[t]{0.475\textwidth}
\centering
\includegraphics[width=1\linewidth]{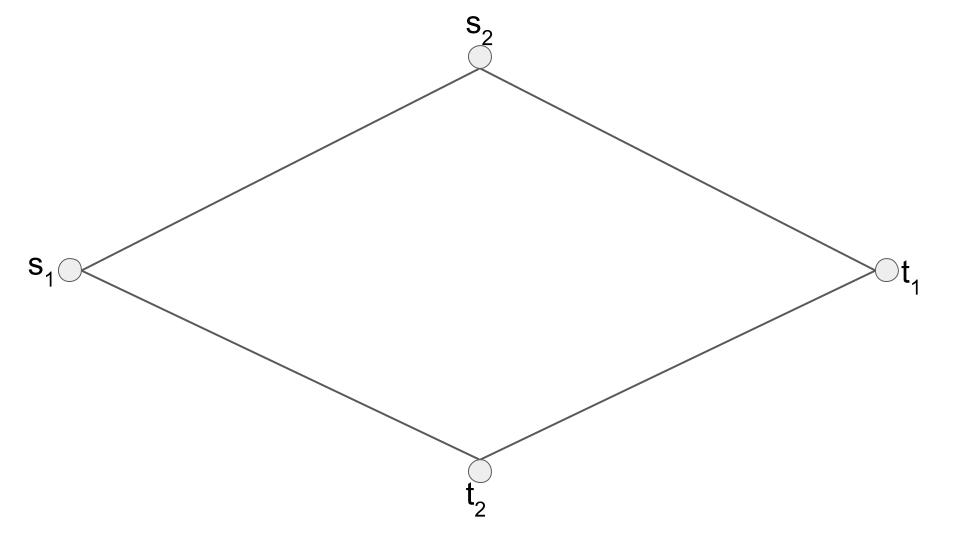}
\caption{Graph $G_1$ and $G_2$.}
\label{fig:g1}
\end{minipage}
\hfill
\begin{minipage}[t]{0.475\textwidth}
\centering
\includegraphics[width=1\linewidth]{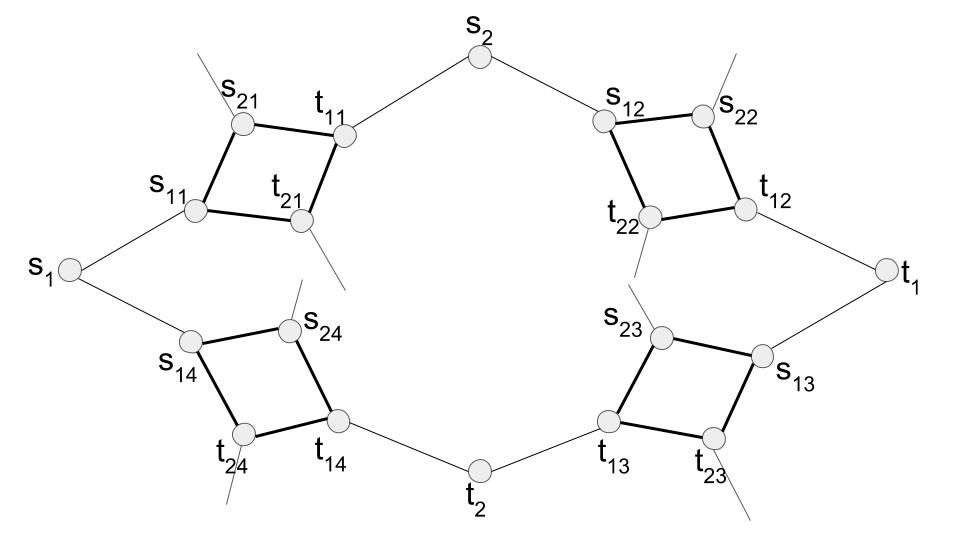}
\caption{Basic gadget that embeds a copy of $G_2$ into each edge of $G_1$. The edges coming out of the copies will be used to connect to other copies of the outer graph ($G_1$). Labelled source-sink pairs of $G_2$ are just for reference. The thin edges are just an artifact and their respective end points represent a single vertex. }
\label{fig:gadget}
\end{minipage}
\end{figure}

In $G_2$, network coding allows us to send more information from every source to its corresponding sink than what simple flows allow. We construct a gadget for the graph tensor exploiting this fact as in Figure \ref{fig:gadget}. We replace each edge of $G_1$ by a copy of $G_2$ with endpoints at a deterministic source-sink pair. We keep the source-sink pairs of $G_1$ and edges of $G_2$. 

For simplicity, assume that each edge in Figure~\ref{fig:g1} has capacity 1. Then the effective capacity at each edge seen by $G_1$ under network coding is more than that under flows. Intuitively, replacing each edge with a source-sink pair of $G_2$ should give network coding a ``capacity advantage'' of $(1+\epsilon_2)$ over multicommodity flow. 
Since the information transferred grows linearly with the capacity, the new information exchanged between source-sink pairs in the gadget under network coding should be $(1+\epsilon_1)(1+\epsilon_2)$ times the information exchanged under flows. 


\begin{figure}
\begin{subfigure}[t]{0.475\textwidth}
\centering
\includegraphics[width=1\linewidth]{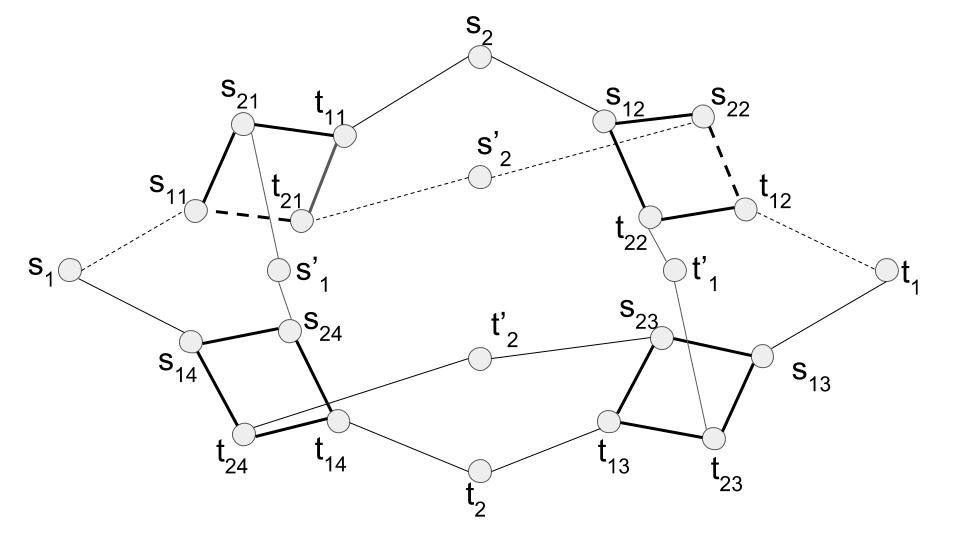}
\caption{A cheating path is highlighted with dashed edges.}
\label{fig:cheat}
\end{subfigure}
\hfill
\begin{subfigure}[t]{0.475\textwidth}
\centering
\includegraphics[width=1\linewidth]{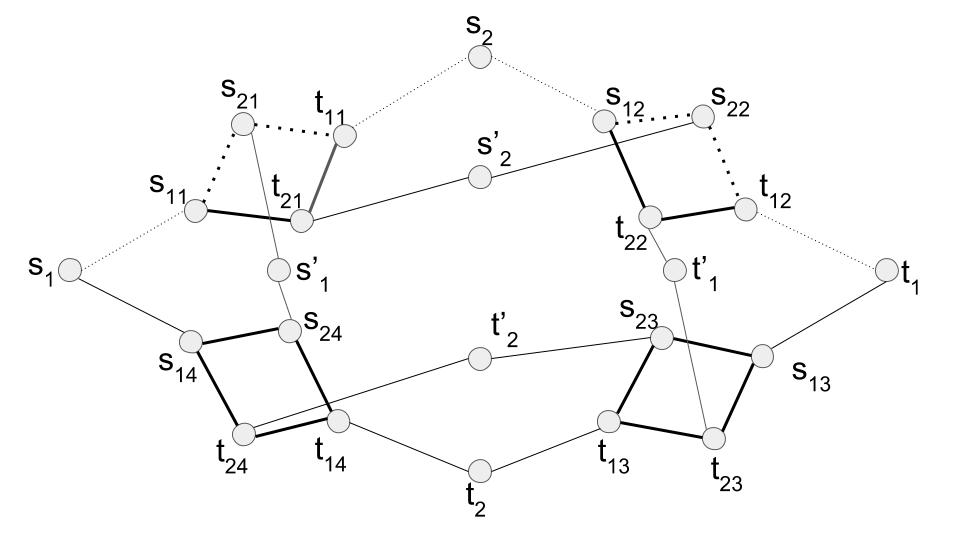}
\caption{An honest path is highlighted with dotted edges.}
\label{fig:honest}
\end{subfigure}
\caption{Two copies of $G_1$ (with source-sink pairs $s_1-t_1,s_2-t_2$ and $s'_1-t'_1,s'_2-t'_2$, respectively) have $8$ edges. These edges are replaced with $4$ copies of $G_2$, each copy having two source-sink pairs. All edges in the first copy of $G_1$ are replaced with $s_{1X}$-$t_{1X}$ source-sink pairs of $G_2$; edges in the second copy of $G_2$ are replaced with $s_{2X}$-$t_{2X}$ source-sink pairs of $G_2$}
\label{fig:tensor}
\end{figure}


We need to be careful because $G_2$ exhibits a gap only when we need to send information from all sources simultaneously. We address this by adding more copies of $G_1$ to the graph tensor and replacing its edges with other source-sink pairs of the copies of $G_2$ as in Figure~\ref{fig:tensor}. 
In each copy of $G_1$ we replace all edges with the same source-sink pair of $G_2$. At the same time, each copy of $G_2$ serves to replace the same edge in all copies of $G_1$.  This is done to facilitate the proof of the upper bound on the MCF rate in the resulting graph. 

Our work is not done here. It is easy to get a lower bound on the network coding rate on the final graph $G$ by just showing a network coding solution. For this, informally, we just compose the network coding solutions for $G_1$ and $G_2$. The hard part is to get an upper bound on the multicommodity flow rate. Since MCF is a linear program, we can upper bound  the value of multicommodity flows by looking at the dual solution of its relaxed linear program. This dual, described in Section~\ref{sec:definitions}, involves computing shortest distances between source-sink pairs under some metric. This metric readily tensorizes: we can take the length of an edge in a copy of $G_2$ to 
be the product of the length of that edge in $G_2$ times the length of the edge(s) in $G_1$ this particular copy of $G_2$ is replacing. 
The problem is to get the lengths of  {\em whole paths} to tensorize.

What could go wrong? Consider Figure~\ref{fig:tensor}.
We would ideally want the dotted paths as in Figure~\ref{fig:honest} to be the  shortest path between $S1$ and $T1$ in $G'$,
since its length is the length of the shortest $s_1-t_1$ path in $G_1$ times the length of the shortest $s_{1X}-t_{1X}$ path in $G_2$. 
Unfortunately, during the tensoring operation we inadvertently introduce additional $s_1-t_1$ paths that do not correspond to ``products'' of paths 
from $G_1$ and $G_2$. 
For example, the dashed path in Figure~\ref{fig:cheat} is a ``cheating'' path which can make the distance between $s_1$ and $t_1$ shorter than expected. We deter the use of ``cheating'' paths by increasing the number of hops between different copies of $G_1$ that a path has to take before it reaches the same copy again.  The technical ingredient which prevents such cheating is in the design of the bipartite graph which tells which copy of $G_1$ should use which copies of $G_2$ (and how to connect them). To prevent cheating, the bipartite graph will need to be of high girth.  The crucial part of the construction is thus constructing high girth bipartite graphs while still keeping check on the size so as to get a $O(\log(\text{size})^c)$ ($c<1$) gap when the tensor is applied repeatedly.

\paragraph{Discussion.} A natural question arises: Can we have a tensor construction that starts with a graph $G$ having some gap between the multicommodity flow rate and the network coding rate and outputs a graph $G'$ with gap $\omega(\log(|G'|))$, thus contradicting \eqref{eq:3} and proving the Li and Li conjecture? We address this question with respect to our construction in Section \ref{sec:sparsity}. We show that the MCF vs. Sparsest Cut 
gap tensorizes for our construction, and thus the tensorization process on its own cannot cause the gap to exceed $O(\log |G'|)$. 

At the same time, if one's goal is to prove the conjecture, it might be easier to reach a contradiction to the gap being $(\log|G|)^c$ than to a constant gap.

\section{Preliminaries}
\label{sec:definitions}

In this section we introduce the problems that we will be interested in studying and any relevant notation. Where appropriate, we use the same notation and definitions as \cite{Kleinberg:Networks, KleinbergHR04}. 

When $G=(V,E)$ is a graph, we specify vertex set of $G$ with $V(G)$ when the underlying graph $G$ is not clear from the context. Similarly $E(G)$ represents the edge set for graph $G$. The set $\{1,2,...,n\}$ is represented by $[n]$. $I(G)$ denotes the set of $k$ source-sink pairs $(s_i,t_i),i\in [k], s_i,t_i\in V$. Given a bipartite graph $B=(V_1,V_2,E)$, we denote the left side of the graph by $V_1(B)$ and the right by $V_2(B)$. A bipartite graph is $(r,s)$ bi-regular when each vertex on the left side has degree $r$ whereas each vertex on right side has degree $s$. 

\subsection{Network coding}

\begin{definition}
An instance of the $k$-pairs communication problem consists of 
\begin{itemize}
	\item a graph $G = (V, E)$,
	\item a capacity function $c: E \rightarrow \mathbb{R}^+$, 
	\item a set $I$ of commodities of size $k$, each of which can be described by a triplet of values $(s_i, t_i, d_i)$ corresponding to the source node, the sink node and the demand of commodity $i$. 
\end{itemize}
\end{definition}

In line with \cite{Kleinberg:Networks}, for undirected graphs we consider each edge $e$ as two directed edges $\vec e, \cev e$, whose capacities will be defined later. It will also be convenient to think of source and sink nodes as edges. Therefore, for every source and sink pair $(s_i, t_i)$, we create new nodes $S_i, T_i$ and connect them via single edges to $s_i, t_i$ respectively. These edges are of unbounded capacity and we will refer to these as the source and sink edges respectively. Every source $S_i$ wants to communicate a message to its sink. 

We give the formal definition of a network coding solution in Appendix~\ref{sec:NC}. Let $M_i$ be the set of messages the $i$-th source-sink pair wants to communicate, and $M = \prod_i M_i$. Let $\Delta(e)$ be the alphabet of characters available at edge $e$. Informally, the solution to a network coding problem must specify for each edge $e$ a function $f_e: M \rightarrow \Delta(e)$, which dictates the character transmitted on that edge. The function $f_e$ must be computable from the characters on the incoming edges at the sender end point. The message at the source and sink edges of any commodity must agree. 

The network coding rate (henceforth known as coding rate) is the largest value $r$ such that for each source-sink pair at least $r \cdot d_i$ information is transmitted while preserving the capacity constraint on all edges. 

\subsection{Multicommodity flow problems and sparsity cuts.}

A flow problem consists 
of a graph $G = (V,E)$ with $k$ commodities 
together with $k$ pairs of nodes 
$(s_i, t_i)$ and quantities $d_i$. The 
goal is to transmit $d_i$ units of 
commodity $i$ from $s_i$ to $t_i$ 
while keeping the total sum of 
commodities that go through a given 
edge $e$ below its capacity $c(e)$. 
There are many optimization problems 
surrounding this problem. We will focus 
on the following one: what is the 
largest $\lambda$ such that at least 
$\lambda$ fraction of each commodity's 
demand is routed? This is justified 
by assuming that no commodity is 
prioritized over another and that all 
resources are shared. We refer to this 
quantity as the flow 
rate of the graph. There are well-known linear programming formulations for these problems (see LP~\ref{eq:LP2} in Section~\ref{sec:MCFs} in the appendix).
Since we will be interested in providing provable upper bounds to the flow rate, it will suffice to look at the dual of this problem. In particular, we use the variables on the following dual LP to provide upper bounds on the flow rate of the sequence of graphs we create. 
We will refer to the $w_{(u,v)}$ as the weight of edge $(u,v)$ in the dual solution.



\begin{equation}
\label{eq:LP3}
\begin{aligned}
\text{minimize} & \sum_{u,v} w_{(u,v)} c(u,v) & & & \\
\text{subject to} & \sum_{(s_i, t_i)} l(s_i, t_i) d_i \ge  1 & & \textrm{(Distance Constraint)} \\
&\sum_{(u,v) \in p} w_{(u,v)} \geq l(s_i,t_i) & & \forall i \in [k], p \in P_i   \\
&w_{(u,v)} \geq 0 & & \forall (u,v) \in E \land \forall (u,v) = (s_i, t_i) \\ 
\end{aligned}
\end{equation}

This LP introduces a semi-metric on the graph which assigns weights to the edges. $l(s_i,t_i)$ is the shortest distance between $i-$th source-sink pair w.r.t. this metric. The goal is to minimize the weighted length of the edges of the graph while maintaining a certain separation between the source-sink pairs. Zero weight edges can be problematic for our graph tensor since they may reduce the weighted girth of the graph in ways we cannot account for. Our tensor, however, does not produce new zero weight edges. Therefore it suffices for our purposes to show that we can get rid of them at the beginning of the construction.

\begin{lemma}\label{standard}
If $G$ is a graph such that the gap between the flow rate and the coding rate is $(1+\epsilon)$, a new graph $G'$ can be constructed such that the gap does not decrease and all the edge weights in LP~\eqref{eq:LP3} are non-zero.
\end{lemma}
\begin{proof} We defer the proof of this lemma to Section \ref{sec:MCFs} of the appendix.
\end{proof}
Interestingly, this lemma is not true for directed graphs.

%
%

 \section{Construction}

In this section we present the construction of our graph tensor and prove our main results, Theorems~\ref{thm:Main1} and~\ref{thm:Main2}. The construction takes two graphs with small gaps and tensors them in such a way that the resulting graph has a gap equal to the product of the previous gaps. Iteratively tensoring a graph with a small gap with itself will yield our main results. 

Throughout this section, when referring to a graph $G_i=(V_i,E_i),i\in[2]$, $k_i$ is the number of source-sink pair, $v_i$ the number of vertices and $m_i$ the number of edges. The capacity of edge $e \in E_i$ will be denoted by $c_{ie}$. 

\subsection{Overview}\label{subsec:overview}
 As mentioned in Section \ref{sec:intro}, we need a bijection between the graph tensor on $G_1$ and $G_2$ and bipartite graphs. We represent the copies of $G_1$ by numbered nodes on the left side of the bipartite graph (say $B$) and copies of $G_2$ by nodes on the right side of $B$. We add an edge $(i,j)$ in $B$ when an edge in the $i$-th copy of $G_1$ got replaced by the $j$-th copy of $G_2$ aligned at a specific source-sink pair. But this definition of bipartite graph $B$ loses information about which specific edge was replaced with which specific source-sink pair. Thus, we consider a variant of bipartite graphs: \emph{colored bipartite graphs}, which have two colors associated with each edge. We will use the first color to represent the edge that got replaced in a copy of $G_1$ and the other to represent the source-sink pair of $G_2$ that replaced that edge. Thus, edges of $B$ get colored from the set $[m_1]\times[k_2]$. Note that each vertex on the left side has degree $m_1$ and that on right hand side has degree $k_2$. The formal description of colored bipartite graphs and \emph{graph tensor} based on this idea is given in Subsection~\ref{subsec:tensor}.
 
 As discussed in Section \ref{sec:intro}, we can avoid ``cheating'' paths by increasing the number of hops that a dashed path (Figure~\ref{fig:tensor}) needs to take to come back to the same copy of $G_1$. Our first requirement would be for the colored bipartite graph $B$ to have high girth. Lemma~\ref{high girth} states the existence of nearly optimal sized high girth bipartite graphs and Subsection~\ref{subsec:colorbipartite} shows how to construct specific colored bipartite graphs (as in Subsection \ref{subsec:tensor}) of high girth. 

Is having a high girth $B$ sufficient for the number of hops to be large? No. When $G_2$ has two sources at the same vertex, the end points (on source side) of the edges in copies of $G_1$ that these two source-sink pairs replaced will collapse on the same vertex implying that we can move between these copies of $G_1$ instantly without traveling along any edge in the tensored graph. But, we would have travelled two consecutive edges in $B$. To remedy this, we condition on the graph $G_2$ to have all sources and sinks lying on distinct vertices. Note that the length of the cheating paths is defined with respect to the weights of edges in a dual solution. Thus, we cannot just transfer the source/sinks to leaves at the corresponding vertex through infinite capacity edges as they would always get weight $0$ in the dual. In Subsection \ref{subsec:def}, we present a way to modify graph $G_2$ to satisfy the above condition.

The multicommodity flow rate for the tensored graph is upper bounded by constructing a dual solution for it based on dual solutions for graphs $G_1$ and $G_2$. In Subsection \ref{subsec:main}, we show the dual construction and prove that the gap of the tensored graph is the xproduct of the previous gaps given appropriate girth. 

The last subsection of this section contains the details of repeated tensoring to get Theorem \ref{thm:Main1}.

\subsection{Graph Tensor}\label{subsec:tensor}
\begin{definition}\label{coloredb}\textbf{Colored Bipartite Graph:} We define $\mathcal{B}_{n_1,n_2,d_1,d_2,g,q_1,q_2}$ to be the set of bipartite graphs $(V_1,V_2,E)$ with girth $g$, $|V_1|=n_1, |V_2|=n_2$, such that degree of each vertex in $V_1$ and $V_2$ is $d_1$ and $d_2$ respectively and each edge is given a color $l_e$ in $[q_1]\times [q_2]$. Note that $n_1d_1=n_2d_2$.
\end{definition}
\begin{definition}\label{coloredtensor} \textbf{$T(G_1,G_2,B)$} is defined to be the graph tensor on directed graphs $G_1$ and $G_2$ based on the colored bipartite graph $B$. 

\end{definition}

For $T(G_1,G_2, B)$ to be defined, we need $B$ to satisfy the following properties: 
\begin{enumerate}
\item $B\in \mathcal{B}_{n_1,n_2,m_1,k_2,g,m_1,k_2}$ for some $n_1, g\in \{1,2,...\}$. 
\begin{itemize}
\item $G_1$ has $m_1$ edges and $G_2$ has $k_2$ source-sink pairs. Therefore the degrees of each node on left and right side should be $m_1$ and $k_2$, respectively. 
\item As mentioned in Subsection \ref{subsec:overview}, edges must be colored in the set $[m_1]\times[k_2]$.
\end{itemize}
\item  $\forall v\in V_2$, the set $B_v=\{b_e\mid \text{$e$ is incident to $v$ and $l_e=(a_e,b_e)$}\}$ is the complete set $[k_2]$. We want each source-sink pair of a copy of $G_2$ to replace some edge in a copy of $G_1$.
\item $\forall u\in V_1$, the set $A_u=\{a_e\mid$$e$ is incident to $u$ and $l_e=(a_e,b_e)$$\}$ is the complete set $[m_1]$. This ensures that each edge in a copy of $G_1$ is replaced.
\item $\forall v\in V_2$, the set $A_v=\{a_e\mid $ $e$ is incident to $v$ and $l_e=(a_e,b_e)$$\}$ has cardinality $1$. To define capacities in the new tensored graph naturally, we want that each source-sink pairs in a copy of $G_2$ replaces some unique edge in the corresponding copies of $G_1$.
\item  $\forall u\in V_1$, the set $B_u=\{b_e\mid \text{$e$ is incident to $u$ and $l_e=(a_e,b_e)$}\}$ has cardinality $1$. This ensures that each edges in a copy of $G_1$ is replaced by the same source-sink pair in different copies of $G_2$.
\end{enumerate}
We construct the graph $T(G_1,G_2,B)$ as follows:

\begin{itemize}
\item Enumerate the $n_1$ nodes in $V_1(B)$ and $n_2$ nodes in $V_2(B)$: $u^{(1)},u^{(2)},...,u^{(n_1)}$ and $v^{(1)},..,v^{(n_2)}$ respectively. 

\item Enumerate all the edges in $G_1$: $e^{(1)}_{G_1},e^{(2)}_{G_1},...,e^{(m_1)}_{G_1}$.
\item Create $n_1$ copies of $G_1$ (vertices and source-sink pairs) and $n_2$ copies of $G_2$ (vertices and edges). Represent the $j^{th}$ copy of graph $G_i,i\in\{1,2\}$ by $G_i^{(j)}$. Let $u^{(i)}\in V_1(B)$ represent the $i$-th copy of $G_1$ and $v^{(i)}\in V_2(B)$ represent the $i$-th copy of $G_2$.
\item For every edge $e=(u^{(i)},v^{(j)})$ colored $(p,k)$, merge the vertices $a_{G_1^{(i)}}$ and $s_{kG_2^{(j)}}$, and $t_{kG_2^{(j)}}$ and  $b_{G_1^{(i)}}$ in $T(G_1,G_2,B)$. Here, $e^{(p)}_{G_1^{(i)}}=(a_{G_1^{(i)}},b_{G_1^{(i)}})$ is the $p^{th}$ edge in the $i$-th copy of $G_1$ and $(s_{kG_2^{(j)}},t_{kG_2^{(j)}})$ is the $k^{th}$ source-sink pair of the $j^{th}$ copy of $G_2$. Informally, we are replacing each edge in a copy of $G_1$ by a copy of $G_2$ with end points aligned with the $k^{th}$ source-sink pair. Set the capacity of every edge $e'$ in this $j^{th}$ copy of $G_2$ to be $c_{1e^{(p)}_{G_1}}c_{2e'}$. This can be done consistently due to Property (4).
\item  Make all the edges undirected.
\end{itemize}

We define a tensor on directed graphs to allow for composition of network coding solutions of $G_1$ and $G_2$. The direction of an edge in $G_1$ tells us how to align the source-sink pair of $G_2$ on that edge. An example of a tensor is the graph in Figure \ref{fig:tensor}. 

\subsubsection{Standard Forms and Graph Extensions}\label{subsec:def}
Without loss of generality, we assume that for the graph $G$, all the demands $d_i, i\in[|I(G)|]$ are equal. Otherwise, we can just divide the demands into small demands of size $x$ such that $x$ divides all the initial rational demands. As discussed in Subsection \ref{subsec:overview}, we want all sources and sinks to lie on distinct vertices.
For all the dual solutions $D$ that we mention, we assume that $D$ does not contain any zero weight edges. This is justified by Lemma \ref{standard} and the fact that new duals constructed while tensoring, which will be defined later, don't create zero weight edges. We say a graph-dual pair ($G$, $D$) is in \emph{standard form} when all the assumptions above are satisfied.

We now present a construction whose goal is to make all $s_i,t_i, i\in[k]$ lie on distinct vertices.

\begin{definition} \label{distinct} Given a graph $G=(V,E)$ with all demands being equal to $d$, and a dual solution $D$ with $\frac{NC_G}{z(D)}\geq 1+\varepsilon$, $\forall \alpha, 0<\alpha<\varepsilon$, construct a new graph $G_{\alpha}$ such that all $s_i,t_i, i\in[k]$ lie on distinct vertices and $G_\alpha$ has a dual solution $D_{\alpha}(G)$ with $\frac{NC_{G_\alpha}}{z(D_{\alpha}(G))}$ being at least $\frac{1+\varepsilon}{1+\alpha}$. $G_\alpha$ is defined as the \textbf{$\alpha$-Extension of $G$ given $D$}. $z(D)$ is the objective value of dual solution $D$.
\end{definition}

Here, we just move the sources/sinks at a vertex to the leaves of the new edges added at this vertex while keeping edge capacities and dual weights in check. The detailed description of $G_\alpha$ and $D_{\alpha}(G)$ is given in Section \ref{sec:sf} of the Appendix.

\subsubsection{Colored Bipartite Graph Construction}\label{subsec:colorbipartite}

We need small, colored bipartite graphs for every degree and girth to define the graph tensor on any two graphs with gaps. We construct such graphs using biregular bipartite graphs with high girth. The following lemma states the existence of nearly-optimal sized colored bipartite graphs.

\begin{lemma}\label{CBexistence}
$\forall r,s,g\ge 3$, there exists a colored bipartite graph $C_{rsg}\in \mathcal{B}_{n_1,n_2,r,s,2g,r,s}$ with $n_1,n_2\le (9rs)^{g+3}$.
\end{lemma}

\begin{proof} We defer the detailed construction and proof of the next lemma to Section~\ref{sec:cbg} of the Appendix. 
\end{proof}

 
\subsubsection{Gap Amplification}\label{subsec:main}
We are given $G_1$ and $G_2$ in standard form with $G_i$, $i \in [2]$ having gap $(1+\varepsilon_i)$. Let $N_i$ be the optimal network coding solution for $G_i, i\in[2]$. Construct a directed graph $G_1'$ from $G_1$ by replacing each (undirected) edge $e=(u,v)\in E(G_1)$ of capacity $c_{1e}$ with 2 directed edges $(u,v)$ and $(v,u)$ of capacities $c_{1eu}$ and $c_{1ev}$ respectively. Here, $c_{1eu}$ and $c_{1ev}$ are the capacities of edge $e$ used by $N_1$ in the defined directions. Note that  $c_{1eu}+ c_{1ev}\le c_{1e}$. Without loss of generality, assume $c_{1eu}+ c_{1ev}= c_{1e}$, as we can always increase one of the capacities without changing the network coding solution to get the equality. Similarly, construct $G_2'$ from $G_2$ based on $N_2$. $G_1'$ and $G_2'$ has $m_1'=2m_1$ and $m_2'=2m_2$ edges respectively.

\begin{definition}\label{tensor}
\textbf{Tensor($G_1,G_2,D_1,D_2$)} is defined as $T(G_1',G_2',B')$, where $B'=C_{m_1'k_2g}$, $2g=\frac{2l_2l_1}{w_1w_2}$. Here, $l_1$ and $l_2$ are the maximum dual distances between any source-sink pair in the dual solutions $D_1$ of $G_1$ and $D_2$ of $G_2$ respectively. $w_1>0$ and $w_2>0$ are the minimum edge weights in the dual $D_1$ and $D_2$ respectively. 
\end{definition} 
Define \textbf{Dual$(G_1,G_2,D_1,D_2)$} to be the specific dual solution for Tensor($G_1,G_2,D_1,D_2$) that would be constructed in proof of Lemma \ref{DualG}. 

All the demands in the graph Tensor($G_1,G_2,D_1,D_2$) are equal to $\frac{d_1d_2}{q}$ where $q=\frac{V_1(B')}{k_2}=\frac{V_2(B')}{m_1}$. Here, $d_i$ is the demand of each commodity in graph $G_i, i\in[2]$. 
We use such a scaling to have a simple description of Dual$(G_1,G_2,D_1,D_2)$ in terms of $D_1$ and $D_2$. 

We prove the gap amplification part of Theorem \ref{thm:Main2} next. The details of how size grows are in the next subsection.

\begin{theorem}\label{gap}
Given graphs $G_1$ and $G_2$ in standard form and dual solutions $D_1$ and $D_2$ respectively, such that $\frac{NC_{G_i}}{z(D_i)}\ge 1+\varepsilon_i, i\in[2]$, $G=\text{Tensor}(G_1,G_2,D_1,D_2)$ has a dual solution $D= \text{Dual}(G_1,G_2,D_1,D_2)$ such that $\frac{NC_G}{z(D)}\ge(1+\varepsilon_1)(1+\varepsilon_2)$.
\end{theorem}
In the next two lemmas, we lower bound the network coding rate and upper bound the multicommodity flow rate of $G$.
\begin{lemma}\label{NCCal} The coding rate for $G$ is at least $r_1r_2(1+\varepsilon_1)(1+\varepsilon_2)q$ where $r_1$ and $r_2$ are objective values of dual solutions $D_1$ and $D_2$ respectively.
\end{lemma}
\noindent\emph{Proof Sketch}. The proof follows from composing the optimal network coding solutions of $G_1$ and $G_2$. The details are given in Section \ref{sec:gaproofs} of the Appendix.

\begin{lemma}\label{DualG} $D$ has objective value at most $r_1r_2q$ where $r_1$ and $r_2$ are the objective values of dual solutions $D_1$ and $D_2$ respectively.
\end{lemma}
\noindent\emph{Proof Sketch}. $G=T(G_1',G_2',B')$ where variables are as defined in Definition \ref{tensor}. For every edge $e\in E(G)$, $e$ is the undirected version of an edge in a copy of $G_{2}'$ (of say $e_2$ in $G_{2}'$) and this copy of $G_{2}'$ must have replaced a unique edge (say $e_1\in E(G_1')$) in different copies of $G_1'$. 
Edges $e_1$ and $e_2$ are directed edges but have undirected counterparts in $G_1$ and $G_2$. Let $w_{1e_1}$ and $w_{2e_2}$ be the weights given to the counterpart edges of $e_1$ and $e_2$ in dual solutions $D_1$ and $D_2$ respectively. Give weight $w_e=w_{1e_1}w_{2e_2}$ to edge $e$ in $D$. Note that $\forall e, w_e>0$ if $w_{1e_1},w_{2e_2}>0\forall e_1,e_2$. Thus, non-zero dual solutions $D_1$ and $D_2$ give a non-zero dual solution $D$ to graph $G$. We still need to show that $D$ is a valid dual solution for $G$. Since $B'$ has girth at least $\frac{2l_2l_1}{w_1w_2}$ and $G_2$ is in standard form, the dotted paths (as in Figure \ref{fig:tensor}) are the shortest paths with respect to dual $D$. We can then write the distances between source-sink pairs in $G$ in terms of the distance of this source-sink pair in $G_1$ w.r.t. $D_1$ and the distance of the source-sink pair in $G_2$ that replaced edges in this copy of $G_1$ w.r.t. $D_2$. This allows us to easily show the satisfiability of the distance constraint for $D$ when demands are as specified in Definition~\ref{tensor}.

The detailed proof is given in Section \ref{sec:gaproofs} of the Appendix. There we also show $z(D)=\frac{n_1}{k_2}z(D_1)z(D_2)=qr_1r_2$.

\begin{proof} \textbf{of Theorem \ref{gap}}: It follows from just dividing the lower bound on the network coding rate of $G$ and the upper bound on the objective value of $D$ obtained in Lemma \ref{NCCal} and Lemma \ref{DualG}.\end{proof}

In the next subsection, we show how to repeatedly apply this construction. Note that, we can only apply the tensor construction on graphs in \emph{standard form}. The following lemma allows us to tensor the new graph obtained with itself. 
\begin{lemma}\label{still standard} Given $G_1$ and $G_2$ in standard form, $\text{Tensor}(G_1,G_2,D_1,D_2)$ is also in standard form.
\end{lemma}
\begin{proof} We defer the proof of this lemma to Section~\ref{sec:gaproofs} of the Appendix. 
\end{proof}
\subsubsection{Iterative Tensoring}\label{subsec:iteration}

In the next two statements size refers to the number of vertices in the graph $A_i$. The calculation of the size involves calculating the required girth at each iteration and the size of the colored bipartite graph used to tensor at each iteration. 

\begin{theorem}
\label{thm:iterative}
Given a graph $A=(V,E)$ with gap $(1+\varepsilon)$, we can construct a sequence of graphs $A_i=(V_i,E_i)$ with gap at least $(1+\frac{\varepsilon}{2})^{2^i}$, size at most $(3c_m)^{(4c_1)^{2^{i+1}}}$ where $c_m$ and $c_1$ are absolute constants. 
\end{theorem}
\begin{proof} We defer the proof to Section~\ref{sec:iterative} of the Appendix. Let $\alpha=\frac{1+\varepsilon}{1+\varepsilon/2}-1$. The proof first considers the $\alpha$-Extension of $A$ to start the recursion with a graph in standard form, then recursively defines pairs of tensored graphs and duals $(A_i, D_i)$ such that the gap increases geometrically. 
\end{proof}


\begin{proof} \textbf{of Theorem \ref{thm:Main1}}: Now, we calculate an expression for the gap in terms of size. $\frac{\log(\text{gap})}{\log(1+\varepsilon/2)}\ge 2^i\ge\frac{ \log\log(\text{size})-\log\log{3c_m}}{\log(16c_1^2)}$. Thus, we get a sequence of graphs with gap at least $\Omega((\log(\text{size}))^{c_2})$ where $c_2$ is an absolute positive constant less than 1 equal to $\frac{\log(1+\varepsilon/2)}{\log(16c_1^2)}$.
\end{proof}

\section{Limits of the Construction} \label{sec:sparsity}

\subsection{Sparsity Squares} 

In this subsection, we show that the construction we present can not be used ``as is'' to prove the Li and Li conjecture. The requirement for the underlying bipartite graph to have a high girth seems to heavily contribute to the size of the graph in the next iteration. Can we do better in terms of size to yield a gap of $\omega(\log |G|)$ by choosing a smaller bipartite graph at every iteration while still having a clever upper bound on the multicommodity flow in the new graph? The answer is no. Theorem~\ref{sparsitysquares} states that for every colored bipartite graph $B$, the tensor of $G_1$ and $G_2$ with $B$ as basis has sparsity of at least the product of the sparsities of $G_1$ and $G_2$ when the demands are all 1 in all the graphs. With the appropriate demands, 
this means that the sparsity grows exactly like the coding rate as in Lemma~\ref{NCCal}. Thus, for any iterative tensoring procedure that starts with a graph $G$ with NC/MCF gap and repeatedly tensors the graph at the $i$th iteration ($G_i$) with itself or with $G$ based on a colored bipartite graph $B_i$ will end up with a graph $G'$ with $\omega(\log |G'|)$ gap. Hence we can start with a graph $H$ with a gap between the flow rate and the sparsity and apply this procedure to get a graph $H'$ with $\omega(\log |H'|)$ gap between the flow rate and the sparsity, contradicting the bounds from~\cite{LeightonRao:MCF}. This means that through iterative tensoring, if we were able to prove the conjecture, we would also prove the statement that there exists no graphs with sparsity-multicommodity flow rate gap which is clearly false. 

\begin{theorem}\label{sparsitysquares}
For any $G_1, G_2, B$ for which $G=T(G_1',G_2',B)$ is defined ($G_1'$ and $G_2'$ are directed graphs obtained from $G_1$ and $G_2$ by directing each edge arbitrary in two directions such that new capacities add up to the previous),  
$$\text{Sparsity}(G) \geq \text{Sparsity}(G_1) \cdot \text{Sparsity} (G_2).$$ when the demands of $G_1$, $G_2$ and $G$ are all scaled to 1.  
\end{theorem}
\begin{proof}
We defer the proof to Section~\ref{sec:sparsityproof} in the Appendix. 
\end{proof}


\bibliographystyle{ACM-Reference-Format-Journals}
\bibliography{citations}
\appendix
\appendix
\section{Definition of Network Coding}
\label{sec:NC}

Let $M(i)$ be the set of all messages $s_i$ wants to send, and let $M = \Pi_i M(i)$. For every $v \in V$, let $\text{In}(v) \subseteq E$ denote the set of edges incident to $e$. 

\begin{definition}
A \emph{network coding solution} for a graph $G$ specifies for each edge directed $e \in E$ an alphabet $\Gamma(e)$ and a function $f_e: M \rightarrow \Gamma(e)$ specifying the symbol 
transmitted on edge $e$. This must satisfy the following two conditions: 	
\begin{itemize}
	\item Correctness: each sink node receives the message from its corresponding source, i.e. $f_{T(i)} = f_{S(i)}$. 
	\item Causality: every message transmitted on edge $e$ is computable from information received at its tail vertex at a time prior to the message's transmission. 
\end{itemize}
\end{definition}

\begin{definition} 
A \emph{causal computation} of a network consists of 
\begin{itemize}
	\item A sequence of edges $e_1, ... , e_T$ where each edge can appear multiple times. 
	\item A sequence of alphabets $\Lambda_1, ... , \Lambda_T$. 
	\item A sequence of coding functions $\rho_1, ... , \rho_T$, which in turn satisfy
	\begin{enumerate}
		\item For each function $\rho_t$ such that $e_t = (u,v)$ is not a source edge, the value of $\rho_t$ is uniquely determined by the values of the functions in the set $\{\rho_x: x < t, e_t \in \text{In}(u) \}$. 
		\item For each edge $e$, the Cartesian product of the alphabets in the set $\{ \Lambda_i : e_i = e\}$ is equal to $\Gamma(e)$. 
		\item For each edge $e$, the set of coding functions $\{\rho_i: e_i = e\}$ together define the coding function $f_e$ specified by the network coding solution. 
	\end{enumerate}
\end{itemize}
\end{definition}

At this point we are equipped with the tools needed to define the network coding rate, the information-theoretic equivalent of the flow rate. 

\begin{definition}
A \emph{network coding solution} for a graph $G$ achieves a rate $r$ if there exists a constant $b \geq 0$ such that 
\begin{itemize}
	\item $H(S(i)) \geq r \cdot d_i \cdot b$ for each commodity $i$ 
	\item for each edge $e \in E$, $H(\vec e)+H(\cev e) \leq c(e) \cdot b$, 
\end{itemize}  
where by $H(\vec e)$ we denote the entropy of edge $\vec e$. The coding rate is defined to be the supremum of the rates of all network coding solutions. 
\end{definition}

\section{Multicommodity Flows} 
\label{sec:MCFs} 

The standard LP formulation for concurrent multicommodity flow problems is written below. It has a variable for every path $p \in P_i$, where $P_i$ is the set of all paths between $s_i$ and $t_i$. We want to find the largest rate $\lambda$ that can be concurrently sent between all source-sink pairs subject to the path variables being non-negative and not exceeding the capacity of any edge over all commodities. 

\begin{equation}
\label{eq:LP2}
\begin{aligned}
\text{maximize} & & \lambda & \\
\text{subject to} & & \sum_{p \in P_{i}} f(p) \geq \lambda d_i & & \forall i \in [k] \\ 
& & \sum_{p:e \in p} f(p) \leq c(e) & & \forall e \in E \\
& & f(p) \geq 0 & & \forall p \\
& & \lambda \geq 0 & & \\
\end{aligned}
\end{equation} 

\begin{proof} \textbf{of Lemma~\ref{standard}}: 
We contract all the edges with zero weight in the dual.
We need to show that the gap does not decrease. Removing a zero dual variable from a multicommodity solution cannot improve the flow rate, since the distances and the dual objective remains the same. We can use the same coding solution for the new graph with the exception that we now compose the encoding on the edges that were contracted. This shows that the flow rate does not increase and the coding rate does not decrease, proving that their ratio does not decrease.
\end{proof}

\section{Standard Form}
\label{sec:sf}
This section gives the detailed description of $G_\alpha$ and $D_\alpha(G)$. 
Let $k_v$ be the number of sources and sinks at vertex $v\in V(G)$. In the graph $G_\alpha$, add $k_v$ edges (leaves) at $v$ with capacity $z(D)d(1+\varepsilon)$ and shift all the sources or sinks at $v$ to the unique endpoints of these leaves. As each source sends $\ge z(D)d(1+\varepsilon)$ amount of information in an optimal network coding solution and can still send $ z(D)d(1+\varepsilon)$, the network coding rate doesn't decrease below $z(D)(1+\varepsilon)$. We construct $D_\alpha(G)$ as follows:
\begin{enumerate}
\item For each edge originally in $G$, assign the same weights as in $D$.
\item Give weight $\frac{\alpha}{kd(1+\varepsilon)}$ to the new edges.
\end{enumerate}
Distances in the dual don't decrease, so $D_\alpha(G)$ is a valid solution. Since we added $k$ new edges, $z(D_\alpha(G))=kz(D)d(1+\varepsilon)\frac{\alpha}{kd(1+\varepsilon)}+z(D)=z(D)(1+\alpha)$. Thus, $\frac{NC_G}{z(D\alpha(G))} \geq \frac{1+\varepsilon}{1+\alpha}$.

\section{Colored Bipartite Graph Construction}
\label{sec:cbg}

\begin{figure}
\begin{minipage}[t]{0.475\textwidth}
\centering
\includegraphics[width=0.35\linewidth]{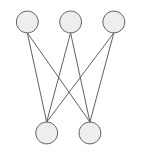}
\caption{(2,3) Bi-regular bipartite graph with girth 4.}
\label{fig:toy}
\end{minipage}
\hfill
\begin{minipage}[t]{0.475\textwidth}
\centering
\includegraphics[width=0.6\linewidth]{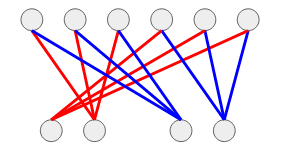}
\caption{Intermediate graph when one set of colors have been assigned.}
\label{fig:middle}
\end{minipage}
\end{figure}

\begin{figure}
\centering
\includegraphics[width=0.8\linewidth]{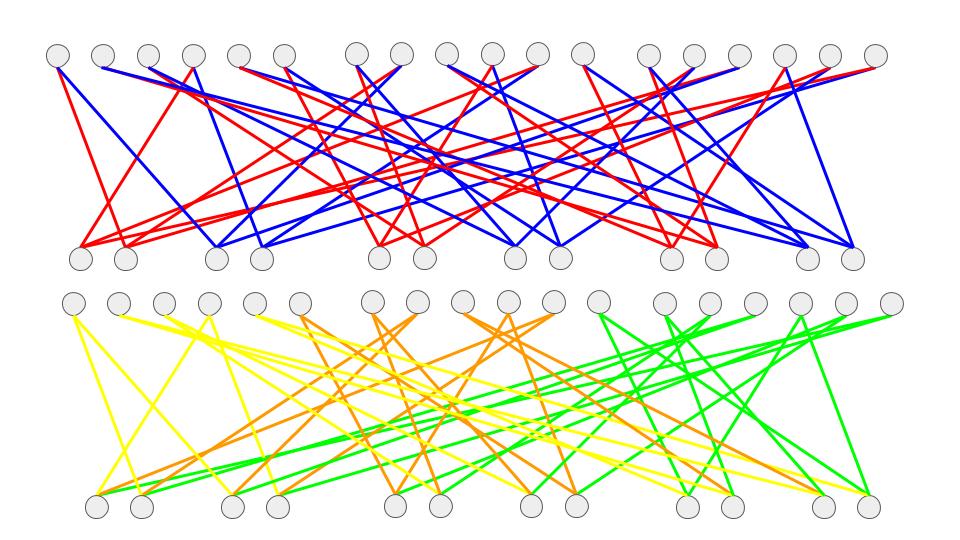}
\includegraphics[width=0.8\linewidth]{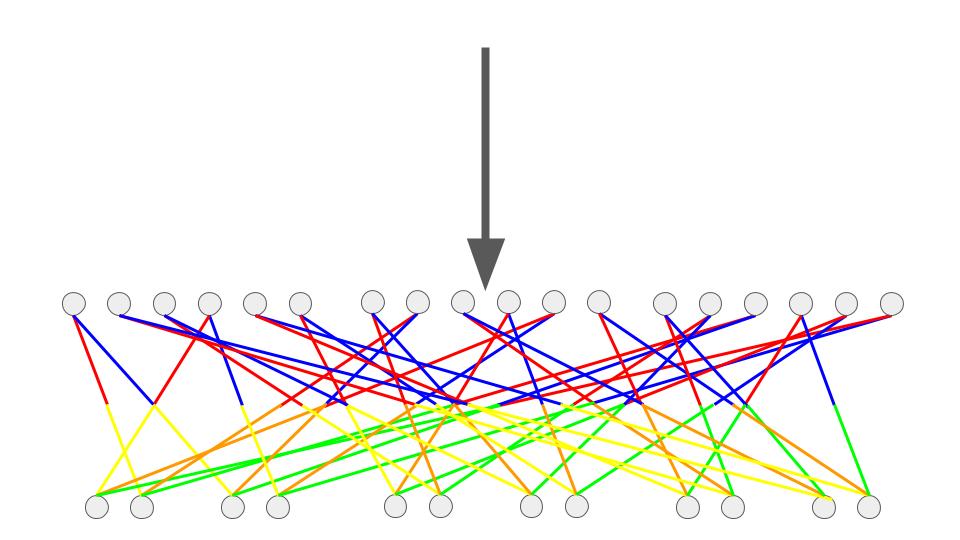}
\caption{Final colored bipartite graph with girth 4.}
\label{fig:final}
\end{figure}
In this section, we give a construction for a colored bipartite graph $C_B$ in $\mathcal{B}_{n_1,n_2,r,s,2g,r,s}$ $\forall r,s,g\ge 3$ with $n_1,n_2\le rs(9rs)^{g+2}$. We start with  $(r,s)$ biregular bipartite graph with girth at least $2g$.

\begin{lemma}\label{high girth} ~\cite{furedi1995} $\forall r,s, g\ge 3$, there exists a $(r,s)$ bi-regular bipartite graph with girth at least $2g$ and having at most $n=(9rs)^{g+2}$ vertices.
\end{lemma}
This lemma follows from Theorem E in ~\cite{furedi1995}.

\begin{proof}\textbf{of Lemma~\ref{CBexistence}.}:
Let $B(r,s,g)$ be a graph satisfying the above property. 
For simplicity, denote $B(r,s,g)$ by just $B=(V_1,V_2,E)$.
Denote the coloring for every edge $e$ by $(a_e,b_e)$. First we construct an intermediate graph $H$ in $\mathcal{B}_{n_1',n_2',r,s,2g,1,s}$ ($n_1'=s|V_1|, n_2'=s|V_2|$) as follows:
\begin{enumerate}
\item Enumerate all the edges incident to a vertex $v\in V_2$ as $e^{(1)}_v,...,e^{(s)}_v$. 
\item Add $s$ copies of $V$ to graph $H$. Enumerate these copies as $(V_1^{(1)},V_2^{(1)}),...,(V_1^{(s)},V_2^{(s)})$.
\item $\forall v^{j}\in V_2^{(j)}, j\in[s]$, $\forall i\in[s]$, corresponding to edge $e^{(i)}_v=(u,v)\in E(B)$, add an edge $e$ from $v^{j}$ to $u^{[(j+i-2 \text{mod } s)+1]}\in V_1^{[(j+i-2) \text{mod } s]+1}$ (copy of $u$ in $[(j+i-2 ) \text{mod } s]+1$-th copy of $V_1$). Set $b_e=[(j+i-2) \text{mod } s]+1$. Therefore, $\forall u^{j}\in V_1^{(j)}$, edges $e'$ incident at $u^j$ have $b_{e'}=j$ (same color).
\end{enumerate}
For a vertex $v\in V_2^{(j)}$, the edge corresponding to $e^{(i)}_v$ comes from a vertex in $V_1^{((j+i-2 \text{mod } s)+1}$. Thus, all edges incident to $v$ have distinct colors. 

We still need to show that the girth of $H$ is at least $2g$. For this, we show that a cycle $C$ of length $c$ in $H$ implies a cycle of length $\leq c$ in $B$. As all the edges incident to a vertex in $H$ correspond to different edges in $B$, when we project back $C$ to a cycle $C'$ in $B$, no two consecutive edges in $C'$ are the same implying $C'$ has no cycle of length 2. Thus, $C'$ must have a cycle of length $3\leq c'\leq c$. $B$ has girth at least $2g$, so the girth of $H$ cannot be smaller.
 
Now, we repeat the process for $H=(H_1,H_2)$ to get graph $C_B$ with $H_1$ playing the role of $V_2$ and $H_2$ playing the role of $V_1$ in the above algorithm. This time we assign $a_e\in[r]$ and make $r$ copies of $H$. We can see that as was the case for $b_e$, each vertex in a copy of $H_1$ gets $r$ distinct $a_e$ values and each vertex in a copy of $H_2$ gets the same $a_e$ depending on which copy it belongs to. The girth doesn't decrease on going from $H$ to $C_B$ giving us the result we claim. 
\end{proof}

An example of a colored bipartite graph in $\mathcal{B}_{12,18,3,2,4,3,2}$ is given in Figure~\ref{fig:final}. We start with $K_{2,3}$ as in Figure~\ref{fig:toy} with girth 4. Then, we construct the intermediate graph in $\mathcal{B}_{4,6,3,2,4,1,2}$ as shown in Figure~\ref{fig:middle}. The color of the edge depends on the copy it is incident to on the lower side. For a vertex on the upper side, we send edges to correct vertices in distinct copies cyclically.

\section{Gap Amplification Proofs}
\label{sec:gaproofs}
\begin{proof} \textbf{of Lemma~\ref{NCCal}}: Graph $G_1$ has a network coding rate of at least $r_1(1+\varepsilon_1)$ and hence each source sends $r_1(1+\varepsilon_1)d_1$ amount of information to its corresponding sink, and similarly for $G_2$. This is true even for the directed graphs $G_1'$ and $G_2'$ by definition. 
While constructing $T(G_1',G_2',B')$, we aligned the source-sink pair in the same direction as the directed edge. This allows us to compose the network coding solutions ($N_2$ over $N_1$) to get the information sent from each source in $G$ to be at least $r_1r_2(1+\varepsilon_1)(1+\varepsilon_2)d_1d_2$. This is due to the fact that as we are replacing edges in $G_1'$ by a source-sink pair of a copy of $G_2'$, the effective capacity seen by the replaced edge $(e)$ with capacity $c_{1e}$ is now $c_{1e}r_2(1+\varepsilon_2)d_2$. Thus, the coding rate for graph $G$ is at least $\frac{r_1 r_2 (1+\varepsilon_1) (1+\varepsilon_2)d_1 d_2}{\text{(demand in graph $G$)}}=\frac{r_1 r_2 (1+\varepsilon_1) (1+\varepsilon_2)d_1 d_2q}{d_1d_2}=r_1r_2(1+\varepsilon_1)(1+\varepsilon_2)q$.
\end{proof}

\begin{proof} \textbf{of Lemma~\ref{DualG}}: Here, we prove that $D$ is indeed a valid dual solution.
$B'$ has $n_1=k_2q$ nodes on the left side. Let $l_1(s_i,t_i)$ denote the shortest distance between $i$-th 
source-sink pair with respect to dual $D_1$. Let $l_2(s_i,t_i)$ denote the shortest distance between $i$-th 
source-sink pair with respect to dual $D_2$.

Let $G_{1}'^u$ and $G_{2}'^u$ be the undirected version of the graphs $G_{1}'$ and $G_2'$ respectively. $G_{1}'^u$ and $G_{2}'^u$ are graphs $G_1$ and $G_2$ where each edge is divided into 2 edges with capacities adding up to the previous one. Construct dual solutions $D_1'$ and $D_2'$ for $G_{1}'^u$ and $G_{2}'^u$ such that each divided edge still gets the same weight as in dual solutions $D_1$ and $D_2$. The distances between source-sink pairs remain the same. In $G$, calculate the shortest distance i.e. $l(s_i^{(y)},t_i^{(y)})$ between source-sink pair $(s_i^{(y)},t_i^{(y)})$ which corresponds to the $i$-th source-sink pair $(s_i,t_i) $ in the $y$-th copy of $G_1'^u$ (finally, we make the graph undirected).  In this copy of $G_1'^u$, assume that we replaced each edge with the $j_y$-th source-sink pair of $G_{2}'^u$ (this is unique due to Property (2) in Definition~\ref{coloredtensor}). Therefore, according to $D$, $l(s_i^{(y)},t_i^{(y)})\le l_1(s_i,t_i)l_{2}(s_{j_y},t_{j_y})$ (these correspond to dotted paths).
Any other path from $s_i^{(y)}$ to $t_i^{(y)}$ involves traversing to another copy of $G_1'^u$ through a copy of $G_{2}'^u$ that replaced edges in this copy of $G_1'^u$. This transition from a copy of $G_1'^u$  to another copy of $G_1'^u$ in $G$ corresponds to two consecutive edges in the bipartite graph $B'$. Any such path in $G$ having no loops would thus have to make at least $g$ of these transitions to revert back to the original copy of $G_1'^u$ containing the source. Here, the girth of graph $B'$ is at least $2g$. Each transition involves crossing at least one edge (in a copy of $G_2'$) with weight at least $w_1w_2$ in $D$ because $G_2$, being in standard form, has all sources and sinks lying on distinct vertices and vertices of a copy of $G_1'$ connect only to the vertices of a copy of $G_2'$ carrying a unique source or sink. Thus, such a path would have distance at least $gw_1w_2=\frac{l_2l_1}{w_1w_2}w_1w_2\ge l_1l_2$ using $l_1,l_2,w_1,w_2$ from Definition~\ref{tensor}. 
The cheating paths have distance at least $l_1l_2$ implying $l(s_i^{(y)},t_i^{(y)})=l_1(s_i,t_i)l_{2}(s_{j_y},t_{j_y})$. The left hand side of the distance constraint in 
LP~\ref{eq:LP3} becomes $\sum_{i=1,y=1}^{k_1,n_1}\frac{d_1d_2}{q}l(s_i^{(y)},t_i^{(y)})$, where the first expression in the summand is the demand of source-sink pairs in $G$. 

\[\sum_{i=1,y=1}^{k_1,n_1}\frac{d_1d_2}{q}l(s_i^{(y)},t_i^{(y)})= \sum_{i=1,y=1}^{k_1,n_1}\frac{d_1d_2}{q}l_1(s_i,t_i)l_{2}(s_{j_y},t_{j_y})= \frac{d_1d_2}{q}\sum_{y=1}^{n_1}l_{2}(s_{j_y},t_{j_y})\sum_{i=1}^{k_1}l_1(s_i,t_i)\]\[=\frac{1}{q}\cdot \frac{n_1}{k_2}\left(\sum_{j=1}^{k_2}d_2l_{2}(s_j,t_j)\right)\left(\sum_{i=1}^{k_1}d_1l_1(s_i,t_i)\right)\ge \frac{n_1}{qk_2}=1 \]

The second to last equality follows from the fact that there are total $n_1$ copies of $G_1'$, $j_y$ is fixed for fixed $y$-th copy of $G_1'$ and each $l_{2}(s_j,t_j)$ ($j\in[k_2]$) is thus counted $\frac{n_1}{k_2}$ time. The last inequality follows from $D_1'$ and $D_2'$ being valid dual solutions of $G_1'^u$ and $G_2'^u$ respectively (distance constraints).

The value of $z(D_1')$ for graph $G_1'^u$ is $r_1$. $D_1'$ assigns the same dual weights as that of $D_1$ for the divided edges and is a valid dual solution for $G_1'$, and similarly for $D_2'$. We can see from the construction of $D$ and the edge capacities that $z(D)=\frac{n_1}{k_2}z(D_1')z(D_2')=qr_1r_2$. $D$ is a function of $G_1$, $G_2$, $D_1$, $D_2$. 
\end{proof}

\begin{proof} \textbf{of Lemma~\ref{still standard}}: Demands are equal for all source-sink pairs in $G=$Tensor($G_1,G_2,D_1,D_2$) by definition. We need to prove that all sources and sinks in $G$ still lie on distinct vertices. We don't add any new source-sink pairs and thus, each source-sink pair lies on distinct vertices on a copy of $G_1'$. While constructing $T(G_1',G_2',B')$, we merge a vertex $v$ in a copy of $G_1'$ with a source or a sink vertex of a copy of $G_2'$ and since each vertex contains a unique source or sink of $G_2'$, no two vertices from different copies of $G_1'$ are merged together. This implies that all sources and sinks still lie on distinct vertices of $G$.
\end{proof}

\section{Proof of Theorem~\ref{thm:iterative}}
\label{sec:iterative}
\begin{proof} Using Lemma~\ref{standard}, we can assume that graph $A$ has an optimal dual solution $D$ with all dual variables being non-zero. It is without loss of generality that $A$ has equal demands for all source-sink pairs. Define $A^*$ to be the $\alpha$-Extension of $A$ given $D$ and $D^*=D_\alpha(A)$ ($1+\alpha=\frac{1+\varepsilon}{1+\varepsilon/2}$). Let $A^*$ have $c_n$ vertices, $c_m$ edges and $c_k$ source-sink pairs having $c_d$ demand each. Without loss of generality we can assume that $c_m\ge c_k, c_n$ as otherwise we can just divide some edges into multiple edges with reduced capacities. Let $l$ be the largest distance between any source-sink pair in the dual $D^*$ and $w>0$ be the minimum weight of an edge in dual $D^*$. We also know that $\frac{NC_{A^*}}{z(D^*)}\ge \frac{1+\varepsilon}{1+\alpha}= 1+\frac{\varepsilon}{2}$. As the objective value of any dual solution is at least the flow rate, we get that $A^*$ has a gap of at least $(1+\frac{\varepsilon}{2})$. $A^*$ is in standard form. $A_i$ is defined iteratively as follows:
\\
\fbox{\begin{minipage}{20em}
\begin{description} 
\item $A_0=A^*, D_0=D^*,\varepsilon_0=\frac{\varepsilon}{2}$.
\item For $i\ge1$:
\item $\varepsilon_i$ is such that $(1+\varepsilon_i)=(1+\varepsilon_{i-1})^{2}$.
\item $D_i=$Dual$(A_{i-1},A_{i-1},D_{i-1},D_{i-1})$.
\item $A_i=$Tensor$(A_{i-1},A_{i-1},D_{i-1},D_{i-1})$.
\end{description}
\end{minipage}}
\\
Note that $\forall i, A_i$ is in standard form using Lemma~\ref{still standard} and thus iterative tensoring is valid. 
Through Theorem~\ref{gap}, we know that if $\frac{NC_{A_{i-1}}}{z(D_{i-1})}\ge (1+\varepsilon_{i-1})$, then $\frac{NC_{A_{i}}}{z(D_{i})}\ge(1+\varepsilon_{i-1})^2=1+\varepsilon_i$. As $\frac{NC_{A^*}}{z(D^*)}=1+\frac{\varepsilon}{2}$, we get $\frac{NC_{A_{i}}}{z(D_{i})}\ge 1+\varepsilon_i=(1+\varepsilon/2)^{2^i}\forall i$ by induction. The objective value of any dual solution is at least the flow rate implying that the gap between coding and flow rate for $A_i$ is at least $(1+\frac{\varepsilon}{2})^{2^i}$. 
%

To see how the size of $A_i$ grows, we first calculate the required girth ($2g_i$) at each iteration. 
From the construction of $D_i= \text{Dual}(A_{i-1},A_{i-1},D_{i-1},D_{i-1})$ in the proof of Lemma~\ref{DualG} we see that $w_i=w_{i-1}^2, l_i\le l_{i-1}^2$. By induction, we have that for all $i$, $w_i= w^{2^i}$ and $l_{i}\le l^{2^i}$,  where $l$ and $w$ are as defined in Subsection~\ref{subsec:iteration} 
From Definition~\ref{tensor}, we have that $g_i=\frac{l_{i-1}^2}{w_{i-1}^2}\le \frac{(l^{2^{i-1}})^2}{(w^{2^{i-1}})^2}=(\frac{l}{w})^{2^i}$. Therefore, $g_i\le (\frac{l}{w})^{2^i}\forall i\ge 1$. Let $c=\frac{l}{w}\ge 1$. 

Now, we establish an upper bound on the size of the graph. Recall $A_i$ is the $T(A_{i-1}',A_{i-1}',B_i)$ where $B_i=C_{m_{i-1}'k_{i-1}g_i}$ and $m_{i-1}'=2m_{i-1}$. $A_{i-1}'$ is the directed graph constructed according to the optimal network coding solution of $A_{i-1}$. Let $n_{1i}=|V_i(B_i)|, n_{2i}=|V_2(B_i)|$. From Lemma~\ref{CBexistence}, $n_{1i}\le (9m_{i-1}k_{i-1})^{g_i+3}\le (9m_{i-1}k_{i-1})^{c^{2^i}+3}$. 

Note that $m_i=\frac{n_{1i}}{k_{i-1}}m_{A_{i-1}'}m_{A_{i-1}'}=\frac{n_{1i}}{k_{i-1}}(4m_{i-1}^2)$ and $k_i=n_{1i}k_{i-1}$. Each edge in $A_{i-1}'$ is replaced by a copy of $A_{i-1}'$ and each copy is counted $k_{i-1}$ times implying $v_i\le 2m_{i-1}v_{i-1}\frac{n_{1i}}{k_{i-1}}$.

Moreover, $\frac{m_i}{k_i}=4(\frac{m_{i-1}}{k_{i-1}})^2$. By induction, $ k_{i}\le m_{i}$ as $c_k\le c_m$. Likewise, we get that $\frac{m_i}{v_i}=2\frac{m_{i-1}}{v_{i-1}}\ge 1 \forall i$. \\
$m_i\le 4n_i(m_{i-1}^2)\le 4m_{i-1}^2(9m_{i-1}k_{i-1})^{g_i+3}\le (9m_{i-1}^2)^{g_i+4}=(3m_{i-1})^{2c^{2^i}+8}\le(3m_{i-1})^{2(c+1)^{2^i}+8}\le(3m_{i-1})^{4(c+1)^{2^i}} \forall i\ge 1$ ($c\ge 1$). 

Let $c_1=c+1$.
\begin{claim} $m_i\le (3c_m)^{(4c_1)^{2^{i+1}}}$.
\end{claim}
\proof For $i=0$, the right hand side evaluates to $(3c_m)^{(4c_1)^2}\ge c_m$, which is equal to the left hand side. Now we assume that the statement is true for $i-1$ and prove for $i$ where $i\ge 1$. $m_i\le (3m_{i-1})^{4c_1^{2^i}}\le (3(3c_m)^{(4c_1)^{2^{i}}})^{4c_1^{2^i}}=3^{(4c_1)^{2^{i}}4c_1^{2^i}+4c_1^{2^i}}c_m^{(4c_1)^{2^{i}}4c_1^{2^i}}\le  (3c_m)^{(4c_1)^{2^{i+1}}}$ as $4^{2^i+1}+4\le 4^{2^{i+1}}\forall i\ge 1$.\\

We have $v_i\le m_i$. Thus, the size of graph $A_i$ is at most $(3c_m)^{(4c_1)^{2^{i+1}}}$. 
\end{proof}

\section{Proof of Theorem~\ref{sparsitysquares}}
\label{sec:sparsityproof}

\begin{proof}
Think of $G_1$ and $G_2$ as undirected $G_1'$ and $G_2'$; their sparsity remains the same. Let $H$ be the set of edges on the cut that achieves the sparsest cut on $G$ separating $n$ source-sink pairs. Consider partitioning this set into sets $H_i = \{ e_{1i}, e_{2i}, ..., e_{h_ii} \}$ according to which copy of $G_2$ (or equivalently $G_2'$), the edge belongs to in $G$. $H_i$ denotes the edges belonging to the $i$-th copy of $G_2$, $|H_i|=h_i$. Note that $|H|=\sum_ih_i$. Let $n_{i}^{(2)}$ be the number of source and sink pairs that $H_i$ separates in the $i$-th copy of $G_2$. These cuts have capacity $\sum_{e \in H_i} c_{2e}$ in $G_2$. By construction, each of these source-sink pairs would have replaced an edge in some copy of $G_1$ (or equivalently undirected $G_1'$). Assume the $k$-th  ($k\in[n_{i}^{(2)}]$) 
source-sink pair replaced edge $e_{i}$ in the $j_{ik}$-th copy of $G_1$ (All source-sink pairs replace the same edge). Mark this edge in the $j_{ik}$-th copy of $G_1$ (which has now been replaced in $G$). The $i$-th copy of $G_2$ makes $n_{i}^{(2)}$ marks. Let $F_j$ be the set of all such marked edges in the $j$-th copy of $G_1$. Let $F_j$ cut $n_{j}^{(1)}$ source-sink pairs in $G_1$. Any source-sink pair that gets cut in $G$ by $H$ must be cut in $G_1$ under $F_j$ by construction. Therefore, $\sum_j n_{j}^{(1)}\ge n$. It is not an equality because there could be a source-sink pair that gets cut by $F_j$ but not by $H$ in $G$, due to paths that travel from the source to other copies of $G_1$ through connecting copies of $G_2$ and come back at the sink. The theorem follows from the following inequalities: 

\begin{equation}
\begin{aligned}
\sum_{e\in H}c_e & = \sum_ic_{1e_{i}}\sum_{e\in H_i}c_{2e}  = \sum_in_{i}^{(2)}c_{1e_{i}}\frac{\sum_{e\in H_i}c_{2e}}{n_{i}^{(2)}}  \\
& \ge \sum_in_{i}^{(2)}c_{1e_{i}}\text{Sparsity} (G_2)  = \text{Sparsity} (G_2)\left(\sum_in_{i}^{(2)}c_{1e_{i}}\right)  \\
& =\text{Sparsity} (G_2)\left(\sum_j\sum_{e\in F_j}c_{1e}\right)  = \text{Sparsity} (G_2)\left(\sum_jn_{j}^{(1)}\frac{\sum_{e\in F_j}c_{1e}}{n_{j}^{(1)}}\right) \\
& \ge  \text{Sparsity} (G_2)\left(\sum_jn_{j}^{(1)}\right)\text{Sparsity}(G_1)  \ge n \left(\text{Sparsity}(G_1) \cdot \text{Sparsity} (G_2)\right) 
\end{aligned}
\end{equation}

The first equality follows from the definition of edge capacities in $G$ in terms of edge capacities in $G_1$ and $G_2$. Since $H_i$ cuts $n_{i}^{(2)}$ source-sink pairs in a copy of $G_2$, the first inequality follows from the Sparsity$(G_2)$ being the smallest ratio for all the cuts. The first equality on the third line follows from the fact that an edge belongs to $F_j$ only when the corresponding source-sink pair that replaced this edge in $G_1$ is cut by the cut corresponding to that copy of $G_2$ and $i$-th copy of $G_2$ result in exactly $n^{(2)}_i$ such edges distributed amongst $F_j$s. Therefore, $\frac{\sum_{e\in H}c_e}{n}\ge \text{Sparsity}(G_1) \cdot \text{Sparsity} (G_2)\implies \text{Sparsity}(G) \geq \text{Sparsity}(G_1) \cdot \text{Sparsity} (G_2)$. 
\end{proof}

\end{document}